\documentclass{article}

\usepackage{arxiv}

\usepackage{amssymb}
\usepackage{amsmath}
\usepackage{amsthm}
\usepackage{tikz}

\usetikzlibrary{calc}
\usepackage{tkz-euclide}
\usetikzlibrary{shapes.geometric}
\usetikzlibrary{automata,positioning}
\usepackage{mathtools}

\newtheorem{lemma}{Lemma}
\newtheorem{theorem}{Theorem}

\newtheorem{proposition}{Proposition}

\usepackage{graphicx}
\usepackage{caption}
\usepackage{subfigure}

\usepackage{colortbl}
\newcommand\LL[1]{\multicolumn{1}{|c}{#1}}
\newcommand\RR[1]{\multicolumn{1}{c|}{#1}}
\newcommand\LR[1]{\multicolumn{1}{|c|}{#1}}
\newcolumntype{x}[1]{>{\centering\arraybackslash\hspace{0pt}}p{#1}}

\usepackage{mathtools}
\usepackage{tikz}
\usepackage{tikz-3dplot}
\usetikzlibrary{arrows,calc,backgrounds}
\usetikzlibrary{decorations.markings}
\usepackage{bbm}

\pgfmathsetmacro{\r}{3} %
\pgfmathsetmacro{\ra}{2}
\pgfmathsetmacro{\h}{21/8} %
\tdplotsetmaincoords{20}{20}

\usepackage{xparse}

\NewDocumentCommand{\twohalfcircle}{ O{#2} m }{%
    \begin{tikzpicture}
    \fill[#2] (0,0) --- (90:0.5ex) arc (90:270:0.5ex) --- cycle;; 
    \fill[#1] (0,0) --- (90:0.5ex) arc (90:-90:0.5ex) --- cycle; 
    \draw (0,0) circle (0.5ex);
    \end{tikzpicture}
}

\title{On the Uniqueness of Nash Equilibria in Multiagent Matrix Games} 

\author{
 James P. Bailey \\
  Industrial and Systems Engineering\\
  Rensselaer Polytechnic Institute\\
  \texttt{bailej6@rpi.edu} 
  }

\date{}

\begin{document}
\maketitle

\begin{abstract}
We provide a complete characterization for uniqueness of equilibria in unconstrained polymatrix games. 
We show that while uniqueness is natural for coordination and general polymatrix games, zero-sum games require that the dimension of the combined strategy space is even. 
Therefore, non-uniqueness is common in zero-sum polymatrix games. 
In addition, we study the impact of non-uniqueness on classical learning dynamics for multiagent systems and show that the classical methods still yield unique estimates even when there is not a unique equilibrium.
\end{abstract}


\section{Introduction}

Aumann describes  two-agent zero-sum games  as ``one of the few areas in game theory, and indeed in the social sciences, where a fairly sharp, unique prediction is made'' \cite{Aumann87}. 
However, despite the general importance of polymatrix (multiagent) games, no uniqueness result is known for the setting and non-degenerate instances of polymatrix zero-sum games have been observed \cite{cai2016zero}. 
In this paper, we study polymatrix games and make three contributions related to the uniqueness of Nash equilibria.
First, we develop a generic method for determining uniqueness of equilibria in classes of polymatrix games.
Second, we provide a complete characterization for uniqueness of equilibria in zero-sum, coordination, and general polymatrix games --- most notably, there are dense, meaningful subclasses of polymatrix zero-sum games that \textit{never have unique equilibria}. 
Finally, we analyze the impact of multiple equilibria in polymatrix zero-sum games on classical learning dynamics and provide insights on recovering unique predictions in these settings --- most notably, we show that the classical gradient descent dynamics still yield a unique prediction in a polymatrix zero-sum game even when there are multiple Nash equilibria thereby recovering the unique prediction heralded by Aumann.

Polymatrix zero-sum games are often used to capture multiagent interactions between decentralized decision makers. 
The outcomes of these systems are characterized by Nash equilibria.
Most theoretical and applied studies of game theory focus around the existence and properties of these equilibria, e.g., \cite{barelli2013note,dasgupta1986existence, carbonell2018existence, fu2022equilibrium}. 
Fortunately for polymatrix games, when strategies are contained in a compact, convex set there always exists at least one equilibrium \cite{howson1972equilibria} and several efficient solution methods are well known, e.g., \cite{deligkas2017computing}. 

However, in many applications of game theory, existence is insufficient and it is important to have a unique equilibrium. 
For instance, the uniqueness of equilibria is considered essential in decentralized security decisions \cite{miura2008security} and the provision of public goods \cite{naghizadeh2017provision}.
Uniqueness of equilibria in polymatrix games also plays an important role in recent applications of online learning. 
The study of online learning to find Nash equilibria in polymatrix zero-sum (and more generally, convex-concave) games has seen a recent surge in attention (e.g., \cite{abernethy2021last,daskalakis2021nearoptimal, Gidel2019Momentum,Mertikopoulos2018CyclesAdverserial}), partially due to its application in Generative Adversarial Neural Networks (GANs) \cite{goodfellow2014generative}. 
Notably, several recent proof approaches for fast convergence of online optimization algorithms in zero-sum games rely on the existence of a unique Nash equilibrium \cite{Bailey2019Regret,bailey2021stochastic,Bailey18Divergence,daskalakis2019last,toonsi2024higher}.

\subsection{Our Contributions}

We develop a roadmap for determining the uniqueness of equilibria for various classes of games; 
specifically, we show that for any class of polymatrix games defined by an analytic matrix function over an open, connected domain that either (a) no game in the class has a unique equilibrium, or (b) almost every game in the class has a unique equilibrium (Theorem \ref{thm:UniqueOrNone}).
This drastically simplifies the process of proving uniqueness of equilibria when studying future classes of polymatrix games; to prove (b), it suffices to eliminate (a) by finding a single game with a unique Nash equilibrium.

We then use Theorem \ref{thm:UniqueOrNone} to characterize uniqueness in zero-sum, coordination, and general polymatrix games with unbounded strategy spaces. 
For coordination and general games, we show that there almost always is a unique equilibrium if and only if the dimension of each single agent's strategy space is at most half of the dimension of the combined strategy space (Theorems \ref{thm:CPerturb} and \ref{thm:GPerturb}), i.e., uniqueness is a fairly natural condition in unconstrained coordination and polymatrix games. 
However, this condition is not sufficient for zero-sum games;
we show that zero-sum polymatrix games have the additional necessary/sufficient condition that the dimension of the combined strategy space is even (Theorem \ref{thm:ZSPerturb}). 
As a result, many instances of polymatrix zero-sum games have multiple equilibria. 

This non-uniqueness is potentially concerning for applications of polymatrix zero-sum games as the sharp, unique prediction that Aumann heralded \cite{Aumann87} is lost in multiagent settings.
As a result, we also study the impact of multiple equilibria on classical learning dynamics for multiagent systems. 
We introduce the \ref{eqn:ContGD} learning dynamics --- a powerful tool commonly used in multiagent systems to update agents strategies and to discover Nash equilibria. 
Perhaps most surprisingly, we show that the existence of multiple Nash equilibria actually \textit{simplifies} the classical learning dynamics; we show that agent strategies, when updated with \ref{eqn:ContGD}, will be contained in an affine subspace of the strategy space that contains a single Nash equilibrium --- the one closest to the initial set of strategies --- as depicted in Figure \ref{fig:MultiEquil} causing the time-average of the dynamics to converge to the unique closest equilibrium
(Theorem \ref{thm:converge2}).
This allows us to recover the primary property of uniqueness that Aumann applauded --- we can make a fairly sharp, unique prediction of the equilibrium obtained by classical learning dynamics in polymatrix zero-sum games. 
Perhaps more importantly, this suggests that recent online learning convergence results \cite{bailey2021stochastic,Bailey18Divergence,Bailey19GDRegret,daskalakis2019last} with proofs that rely on uniqueness of equilibria likely extend even when there is not a unique equilibrium in the system.

\begin{figure}[ht]
\centering
    \begin{tikzpicture}[
    tdplot_main_coords,
    Helpcircle/.style={gray!70!black,
    },scale=0.5
    ]

    \draw[dashed] (0,-4,0) -- (0,8,0) node[below,black] {Set of Nash Equilibria};

    \coordinate (M1) at (0,0,0); 
    \coordinate (M2) at (0,4,0); 
    \coordinate (A) at (0,\h,0);

    \tdplotsetrotatedcoords{90}{90}{0}%
    \draw[thick, Helpcircle, red, tdplot_rotated_coords,decoration={markings, mark=at position 1 with {\pgftransformscale{3}\arrow{>}}},
            postaction={decorate}] (A) circle[radius=sqrt(\r*\r-\h*\h)];

    \begin{scope}[tdplot_screen_coords, on background layer]
    \fill[ball color= gray!20, opacity = 0.25] (M1) circle (\r); 
    \end{scope}

    \begin{scope}[tdplot_screen_coords, on background layer]
    \fill[ball color= gray!20, opacity = 0.25] (M2) circle (\ra); 
    \end{scope}

    \fill[black] ({sqrt(\r*\r-\h*\h)},\h,0) circle[radius=3pt];
    \node[right,black] at  ({sqrt(\r*\r-\h*\h)},\h,0) {$x(0)$};

    \fill[black] (M1) circle[radius=3pt];
    \node[right,black] at  (M1) {$x^*-\lambda \cdot d$};
    
    \fill[black] (M2) circle[radius=3pt];
    \node[right,black] at  (M2) {$x^*+d$};

    \draw[black] (M1)--({sqrt(\r*\r-\h*\h)},\h,0)--(M2);
    
    \fill[black] (A) circle[radius=3pt];
    \node[right,black] at  (A) {${x}^*$};
    \end{tikzpicture}
\caption{Despite non-uniqueness, \ref{eqn:ContGD} contains strategies to a subspace with a unique equilibrium. 
For instance, when using \ref{eqn:ContGD} in $\mathbb{R}^3$, agents' strategies remain equidistant from distinct Nash strategies $x^*+d$ and ${x}^*-\lambda \cdot d$ causing agent strategies to cycle around a lower-dimensional circle in $\{x\in \mathbb{R}^3: d^\intercal x = d^\intercal x^*$\} that uniquely contains the unique Nash equilibrium ${x}^*$  that is closest to the initial strategy $x(0)$.} \label{fig:MultiEquil}
\end{figure}
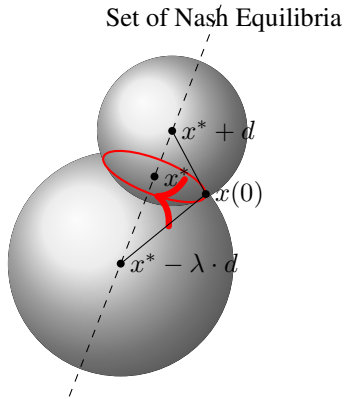



\section{Notation} \label{sec:Notation}

A \textit{polymatrix} game is a non-cooperative game with $n\geq 2$ agents where agent $i\in [n]:=\{1,...,n\}$  selects a $k_i$-dimensional strategy --- specifically, agent $i$ selects a strategy $x_i \in {\mathcal X}_i$ where ${\mathcal X}_i\subseteq \mathbb{R}^{k_i}$ is agent $i$'s strategy space.
For this paper, we focus on unbounded polymatrix games, i.e.,  ${\mathcal X}_i=\mathbb{R}^{k_i}$, a common area of focus in the develop of online optimization algorithms, e.g.,  \cite{Bailey2019Regret,Mertikopoulos2018CyclesAdverserial,pmlr-v202-vyas23a,NEURIPS2023_bf331c87}.

If agent $i$ selects strategy $x_i$ and agent $j$ selects $x_j$, then agent $i$ receives a utility of $x_i^\intercal A^{(ij)}x_j$ where $A^{(ij)}\in \mathbb{R}^{k_i\times k_j}$ is agent $i$'s payoff matrix for interactions with agent $j$. 
Once all agents select their strategies, agent $i$ receives a total utility of $x_i^\intercal (-b_i + \sum_{j\neq i}A^{(ij)}x_j)$ where $b_i$ denotes the (linear) cost for playing strategy $x_i$.

Frequently throughout this paper, we will consolidate agents' payoff matrices into a single matrix in the form
\begin{align}
    A:=\left[
    \begin{array}{c c c c}
    A^{(11)}=\{0\}^{k_1\times k_1}	&	A^{(12)}				&  \dots		&	A^{(1n)}				\\
    A^{(21)}		&	A^{(22)}	&  \dots		&	A^{(2n)}				\\
    \vdots					&	\vdots					&  \ddots		&	\vdots					\\
    A^{(n1)}					&	A^{(n2)}					&  \dots		&	A^{(nn)}	
    \end{array}\right]\tag{Consolidated Payoff Matrix}
\end{align}
and agents' cost vectors into the single vector
\begin{align}
    b:= \left[\begin{array}{c}b_1 \\ \vdots \\ b_n \end{array} \right]. \tag{Consolidated Cost Vector}
\end{align}

    A polymatrix game $G=(n,k,A,b)$ is in the form:
    \begin{align}\tag{Polymatrix Game}\label{eqn:BilinearGame}
	\max_{x_i\in \mathbb{R}^{k_i}}x_i ^\intercal &\left(-b_i + \sum_{j\neq i}A^{(ij)}x_j \right) \ \forall i\in[n].
    \end{align}

A solution to this game, if one exists, is called a Nash equilibrium.
The strategy profile $x^*$ is a Nash equilibrium if, for agent $i$, the strategy $x_i^*$ is a best response to $x_{-i}^*$ where $x_{-i}^*$ denotes the strategies of all agents except agent $i$.  Formally, $x^*=(x_1^*,...,x_n^*)$ is a Nash equilibrium if and only if the following holds for each agent:
\begin{align*}
\ x_i^\intercal \left(-b_i + \sum_{j\neq i}A^{(ij)}x_j^* \right) 
\leq &\  x_i^{*\intercal} \left(-b_i + \sum_{j\neq i}A^{(ij)}x_j^* \right) \ \forall x_i \in \mathbb{R}^{k_i}.\tag{Nash Conditions}\label{eqn:NashConditions}
\end{align*}
Equivalently, if agent $i$ switches from strategy $x_i^*$ to $x_i$, then the quality of agent $i$'s outcome should not improve; the strategy $x_i^*$ is a best response to the combined opponent strategy $x_{-i}^*$. 

The three primary types of polymatrix games examined in the literature are zero-sum games, coordination games, and general games which are defined as follows:
\begin{itemize}
	\item[] \textbf{Zero-Sum:} $A^{(ij)}=-\left(A^{(ji)}\right)^\intercal$ for all $i$ and $j$ --- i.e., the consolidated payoff matrix is $A=-A^\intercal$.
	\item[] \textbf{Coordination:} $A^{(ij)}=\left(A^{(ji)}\right)^\intercal$ for all $i$ and $j$ --- equivalently, $A=A^\intercal$.
	\item[] \textbf{General:} No restrictions on $A^{(ij)}$. 
\end{itemize}

These definitions do not make use of the linear cost $x_i^\intercal  b_i$.  
Games are intended to capture interaction between individuals. 
While $b_i$ does impact i's payoff, it does not capture interaction and is not referenced in the definitions of zero-sum, coordination, and general polymatrix games. 
The role of $b_i$ is discussed further in Section \ref{sec:Constraints}. 

\section{A Roadmap for Determining Unique Equilibria.}\label{sec:roadmap}

In this section, we establish a simple method for determining whether a class of games yields unique equilibria.  
We show that for a class of games defined through an analytic matrix function that either (1) no game in the class yields a unique equilibrium, or (2) almost every game in the class yields a unique equilibrium (Theorem \ref{thm:UniqueOrNone}).  
This provides a very simple method for proving uniqueness in games; to eliminate (1) it suffices to show there exists a single instance in the class with a unique equilibrium. 

We begin with two simple results characterizing when a unique equilibrium exists.

\begin{lemma}\label{lem:Nash} $x^*$ is a Nash equilibrium for the \ref{eqn:BilinearGame} $G=(n,k,A,b)$ if and only if it satisfies
\begin{align}\tag{Simplified Nash Conditions}\label{eqn:NashConditions2}
	\left(-b_i + \sum_{j\neq i}A^{(ij)}x_j^* \right)=\vec{0}^{k_i} \ \forall i\in [n].
\end{align}
where $\vec{0}^{k_i}$ is a vector of $k_i$ 0's.  Equivalently, $x^*$ is a Nash equilibrium if and only if $Ax^*=b$. 
\end{lemma}
\begin{proof}
	$(\Rightarrow)$: We proceed by using the contrapositive and show that $x^*$ is not a Nash equilibrium if $-b_i + \sum_{j\neq i}A^{(ij)}x_j^*=v_i\neq \vec{0}^{k_i}$ for some agent $i$.  Suppose agent $i$ updates their strategy to $x_i^*+v_i$ resulting in the utility,
	\begin{align*}
		&\ (x_i^{*}+v_i)^\intercal  \left(-b_i + \sum_{j\neq i}A^{(ij)}x_j^* \right)\\
  =
  &\ x_i^{*\intercal} \left(-b_i + \sum_{j\neq i}A^{(ij)}x_j^* \right)+||v_i||^2\\
  >&\ x_i^{*\intercal} \left(-b_i + \sum_{j\neq i}A^{(ij)}x_j^* \right)
	\end{align*} 
	implying $x_i^*+v_i$ is a better response to $x_{-i}^*$ than $x_i^*$ and therefore $x^*$ is not a Nash equilibrium. 
	
	$(\Leftarrow):$ If $\left(-b_i + \sum_{j\neq i}A^{(ij)}x_j^* \right)=\vec{0}^{k_i}$, then $x_i^\intercal \left(-b_i + \sum_{j\neq i}A^{(ij)}x_j^* \right)= x_i^\intercal \vec{0}^{k_i}=0$ for all $x_i$ trivially satisfying the \ref{eqn:NashConditions}.
\end{proof}

\begin{lemma}\label{cor:invert}
	A \ref{eqn:BilinearGame} $G(n,k,A,b)$ has a unique Nash equilibrium $x^*$ if and only if $A$ is invertible. Specifically, the unique Nash equilibrium, if one exists, is given by $x^*=A^{-1}b$.
\end{lemma}

In this paper, we aim to characterize when the Nash equilibrium is unique. 
By definition of linear independence, the \ref{eqn:NashConditions} have a unique solution if and only if the consolidated payoff matrix $A$ is invertible. 
Thus, we explore necessary and sufficient conditions for the existence of $A^{-1}$ in polymatrix games.

\subsection{Conditions for the Invertibility of $A$}

It is well-known that the set of non-invertible matrices have measure 0 and, by Lemma \ref{cor:invert}, we might expect that every polymatrix game has a unique Nash equilibrium. 
However, the set of matrices that actually correspond to a \ref{eqn:BilinearGame} also have measure 0 --- i.e., most matrices do not correspond to polymatrix games --- since the block matrices on the diagonal are uniformly zero (i.e., agent $j$ does not play against themselves and $A^{(jj)}=0$), and it is not clear whether every matrix in the class of polymatrix games is invertible. 
We provide a roadmap for identifying when a unique Nash equilibrium is a generic assumption and, in Section \ref{sec:polymatrix}, use that roadmap to characterize the uniqueness of Nash for zero-sum, coordination, and general games.

\begin{theorem}\label{thm:UniqueOrNone}
    Let $K=\sum_{i\in [n]} k_i$ be the total number of strategies and $A: {\mathcal U} \to \mathbb{R}^{K\times K}$  be an analytic matrix-function that generates polymatrix games where ${\mathcal U}$ is a connected open domain of $\mathbb{R}^d$ with probability measure $\mu$. 
	Either (1) no $G=(n,k,A(U),b)$ has a unique Nash equilibrium, or (2) the set of  $U\in {\mathcal U}$ where $G=(n,k,A(U),b)$ has a unique Nash equilibrium has probability measure one. 
\end{theorem}

We remark that the Theorem \ref{thm:UniqueOrNone} also extends to a variety of quasi-analytic functions.
However, for this paper, non-analytic, and even most analytic functions, are not particularly interesting.  
Instead, to explore the full range of coordination, zero-sum, or general games, we only need the identity function $A(U)=U$. 

To prove Theorem \ref{thm:UniqueOrNone}, we make use of a similarly worded statement for the roots of an analytic function. 

\begin{lemma}[Roots of an Analytic Function \cite{mityagin2015zero}]\label{lem:Zeros}
    Let $g(x)$ be an analytic function on an open, connected domain ${\mathcal U}\subseteq \mathbb{R}^d$ with measure $\mu$. 
    Either (1) $g$ is identically zero, or (2) its zero set
        $R(g):= \{x\in {\mathcal U}: g(x)=0\}$
    has zero measure. 
\end{lemma}

\begin{proof}[Proof of Theorem \ref{thm:UniqueOrNone}]
    By Lemma \ref{cor:invert}, $A(U)$ has a unique Nash equilibrium if and only if $A(U)$ is invertible; equivalently, if $det(A(U))\neq 0$. 
    Let $g(U)=det(A(U))$. 
    Since a matrix is invertible if and only if its determinant is non-zero, the statement of Theorem \ref{thm:UniqueOrNone} is equivalent to: Either (1) $g(U)=0$ for all $U\in {\mathcal U}$, or (2) the zero set of $g$ has measure zero.

    Using the Leibniz formula,  $g(U)= \sum_{\sigma \in S_K} sgn(\sigma)\prod_{i\in \sigma}A_{i,\sigma_i}(U)$ where $S_K$ denotes the set of permutations of $[K]$ and $sgn(\sigma)$ denotes the signature of $\sigma$, i.e., $sgn(\sigma)=1$ if $\sigma$ can be obtained from $[K]$ via an even number of pairwise exchanges and $sgn(\sigma)=-1$ otherwise. Therefore $g$ is an analytic function since $A$ is analytic and the determinant is an analytic (polynomial) function. 
    Theorem \ref{thm:UniqueOrNone} then follows by applying Lemma \ref{lem:Zeros} to $g(U)$.  
\end{proof}

	Theorem \ref{thm:UniqueOrNone} provides a simple outline to prove uniqueness. 
	So long as set of games analyzed is generated from an analytic function over a connected, open domain and there is at least one game with a unique Nash equilibrium, then almost every game generated in this way will have a unique Nash equilibrium.
	In this paper, we are solely interested in zero-sum, coordination, and general games. 
	However, if in the future we wish to show some other set of games has a unique Nash equilibrium, then Theorem \ref{thm:UniqueOrNone} provides a simple roadmap. 

\section{Polymatrix Games}\label{sec:polymatrix}

We begin by providing a necessary condition for a unique Nash equilibrium for all games. 
Specifically, we show if the dimension of an agent's strategy space is more than half of the dimension of the combined strategy space, then there is not a unique Nash equilibrium. 

\begin{proposition}\label{thm:Half}
    The polymatrix  game $G=(n,k,A,b)$ does not have a unique Nash equilibrium if there is an agent $i\in [n]$ such that $k_i> K/2=\sum_{j\in [n]}k_j/2$. 
\end{proposition}
\begin{proof}
    By Lemma \ref{cor:invert}, it suffices to show that $A$ is not invertible.  
    Using the Leibniz formula for the determinant, $det(A)=\sum_{\sigma \in S_K} sgn(\sigma)\prod_{w=1}^K A_{w,\sigma_w}$.

 Consider any \ref{eqn:BilinearGame} where the dimension of agent i's strategy space is more than half of the dimension of the combined strategy space, i.e.,  $k_i> K/2$, and consider any permutation, $\sigma$, of $[K]$. 
 Without loss of generality, we assume that $i=1$ so that agent $i$ controls strategies $1,..,k_1$. 
 Through a standard application of the pigeonhole  principle, there is at least one $j \in [k_1]$ such that $\sigma_j\in [k_1]$. 
For that particular $j$, $A_{j,\sigma_j}=0$ by definition of $A$ --- equivalently since no agent is playing against themselves, i.e., since $A^{(ii)}$ is uniformly zero. 
Therefore $\prod_{w=1}^K A_{w,\sigma_w}=0$.  
This holds for all $\sigma\in S_K$ and therefore $det(A)=0$, $A$ is not invertible, and there is not a unique Nash equilibrium by Lemma \ref{cor:invert}. 
\end{proof}

We show later this condition is also sufficient for both coordination and general games, i.e., almost every coordination/general game the dimension of each agent's strategy space is at most half of the dimension of the combined strategy space has a unique Nash equilibrium. 
However, zero-sum games have an additional necessary condition, which we discuss in Section \ref{sec:ZS}.

\subsection{Coordination Games and General Games}\label{sec:CG}

In this section, we use Theorem \ref{thm:UniqueOrNone} to show that almost every coordination and general polymatrix game will have a unique Nash equilibrium so long as the dimension of each agent's strategy space is at most half of the dimension of the combined strategy space. 
To apply Theorem \ref{thm:UniqueOrNone}, we need to define an analytic matrix function over an open, connected domain and a respective probability measure such that every coordination/general polymatrix game can be obtained. 
Then, by Theorem \ref{thm:UniqueOrNone}, either no game has a unique Nash equilibrium or almost every game does.  
Thus, it will suffice to provide a single game obtainable by the matrix function and probability measure that has a unique equilibrium.

\begin{theorem}\label{thm:CPerturb}
    Almost every polymatrix coordination game has a unique Nash equilibrium if and only if $k_i\leq K/2=\sum_{j\in [n]} k_j/2$ for all $i\in [n]$.
    Formally, let ${\cal U}^{(ij)}=\mathbb{R}^{k_i\times k_j}$ for $i<j$, ${\cal U}=\times_{i<j}{\cal U}^{(ij)}$.  
    For $i<j$, let $A^{(ij)}(U)=U^{(ij)}$ and $A^{(ji)}(U)=[U^{(ij)}]^\intercal$. 
    Finally, let $\mu$ be an arbitrary full-dimensional Gaussian distribution over $\cal U$. 
    Then the set of $U\in {\cal U}$ that has a unique Nash equilibrium has probability measure one if and only if $k_i\leq K/2=\sum_{j\in [n]} k_j/2$ for all $i\in [n]$.
\end{theorem}

\begin{proof}
    In this proof, each $b$ and $k$ will correspond to a different class of games. 
    Thus, we show this result for an arbitrary $b\in \mathbb{R}^K$ and for an arbitrary $k\in \mathbb{Z}_{>0}^n$ where $\sum_{i\in [n]} k_i=K$ and $k_i \leq K/2$ for all $i\in [n]$. 

    Necessity is given by Proposition \ref{thm:Half}.
    To establish sufficiency, we make use of Theorem \ref{thm:UniqueOrNone}. 
    By definition of ${\cal U}$, a game with costs $b$ and dimensions $k$ is a coordination game if and only if it can be realized by $(A, {\cal U}, \mu)$. 
    Furthermore, $A(\cdot)$ is an analytic matrix function over the open, connected ${\cal U}$ and we can apply Theorem \ref{thm:UniqueOrNone}.
    By Theorem \ref{thm:UniqueOrNone}, either (1) no realized game has a unique equilibrium, or (2) almost every game does. 
    To eliminate (1) and to prove (2) it suffices to provide a single realizable game that has a unique equilibrium. 
    We now break the proof into cases based on the parity of $K$. 

    \textbf{For $K$ even:}
    \begin{align}\tag{Coordination Game with Unique Nash for Even $K$}
		A=
			\left[
				\begin{array}{x{0.6cm} x{0.6cm} x{0.6cm} x{0.6cm} x{0.6cm} x{0.6cm} }
					\cline{1-3}
					\LL{0} & 0 & \RR{0} & \cellcolor{gray!25}1& 0 & 0\\
					\cline{2-4}
					\LL{0} & \LL{0} & \RR{0} & \RR{0} & \cellcolor{gray!25}1 & 0\\
					\cline{3-5}
					\LL{0} & \LL{0} & \LR{0} & \RR{0} & \RR{0} & \cellcolor{gray!25}1\\
					\cline{1-3}\cline{4-6}
					\cellcolor{gray!25}1 & \LL{0} & \LR{0} & \RR{0} & \RR{0} & \RR{0}\\
					\cline{2-4}
					0 & \cellcolor{gray!25}1 & \LR{0} & 0 & \RR{0} & \RR{0}\\
					\cline{3-5}
					0 & 0 & \RR{\phantom{-}\cellcolor{gray!25}1\phantom{-}} & 0 & 0 & \RR{0}\\
					\cline{4-6}
				\end{array}
			\right]
\end{align}
    We construct a single consolidated payoff matrix that will be valid for any $k$ where $k_i< K/2$ for all $i\in [n]$.
    Let $A_{w,\frac{K}{2}+w}=A_{\frac{K}{2}+w,w}=1$ for all $w\in \left[\frac{K}{2}\right]$ and let $A_{w,v}=0$ for all other pairs $\{w,v\}$ implying $A=A^\intercal$. 
    The matrix $A$ is such that each $K/2$ square submatrix centered on the diagonal of $A$ is filled with zeroes and such that the first sub/superdiagonals below/above these matrices are filled with ones as shown above for $K=6$.

    Once again, the Leibniz formula for the determinant is $det(A)= \sum_{\sigma \in S_K} \prod_{w\in [K]} A_{w,\sigma_w}$. 
    Since every row/column of $A$ contains a single non-zero element, the only permutation of $[K]$ with a non-zero product is $\sigma=[K/2+1,K/2+2,...,K,1,2,K/2]$ as depicted by the gray cells in the above matrix. Therefore $|det(A)|=1$. This implies $A$ is invertible, there is a unique solution to $Ax=b$, and, so long as $A$ corresponds to a polymatrix coordination game, $A$ has a unique equilibrium. 
    
    It remains to verify that $A$ corresponds to a polymatrix coordination game, i.e., that $A^{(ij)}=[A^{(ji)}]^\intercal$ and that $A^{(jj)}=\{0\}^{k_j\times k_j}$.  
    Let $s_i\subset[K]$ and $s_j\subset[K]$ denote the components of $x=[x_1, x_2, ..., x_n]$ that correspond to agents $i$ and $j$ respectively so that  $A^{(ij)}=A_{s_i,s_j}$. 
    By construction, $A=A^\intercal$ and $A^{(ji)}= A_{s_j,s_i}= A_{s_i,s_j}^\intercal=[A^{(ij)}]^\intercal$.

    To show $A^{(jj)}=\{0\}^{k_j\times k_j}$, simply observe that $A^{(jj)}$ is necessarily a square matrix of size at most $K/2$ centered on the diagonal of $A$ and therefore is everywhere $0$ by construction.  
    Formally, suppose agent $j$ controls strategy $w$ in $A$.  
    If $w\leq K/2$, then the only non-zero entry in row $w$ of $A$ is in column $K/2+w$ which is not controlled by agent $j$ since $j$ controls at most $K/2$ strategies implying every entry corresponding to row $w$ in $A^{(jj)}$ is zero. 
    Similarly, if $w=w'+K/2> K/2$, then the only non-zero entry in row $w$ of $A$ is in column $w'=w-K/2$ which is not controlled by agent $j$ since $j$ controls at most $K/2$ strategies and every entry corresponding to row $w$ in $A^{(jj)}$ is zero. 
    Thus, $A^{(jj)}$ is uniformly zero. 

    Therefore $A$ corresponds to a polymatrix coordination game with a unique Nash equilibrium on $n$ agents where agent $i$ controls $k_i$ strategies. 
    Thus, by Theorem \ref{thm:UniqueOrNone}, almost every polymatrix coordination game in this form will have a unique Nash equilibrium.

    \textbf{For $K$ odd:},
we construct a similar matrix.
    As with the previous case, the matrix will contain a square $\lfloor \frac{K}{2} \rfloor$ submatrix of zeroes along the diagonal.
    Unlike the even case, the first two sub/super diagonals below/above the zero submatrices will contain 1s. 
    Formally,
    \begin{align*}
        A_{w,\lfloor\frac{K}{2}\rfloor+w}=A_{\lfloor\frac{K}{2}\rfloor+w,w}=1 \ for\ & w\leq \lceil K/2\rceil \tag{first sub/superdiagonal above zero submatrices}\\
        A_{w,\lceil\frac{K}{2}\rceil+w}= A_{\lceil\frac{K}{2}\rceil+w,w}=1 \ for\ & w\leq \lfloor K/2\rfloor \tag{second sub/superdiagonal above zero submatrices}
    \end{align*}
    and $A_{wv}=0$ for all other $\{w,v\}$ as shown below for $K=7$.
    Through the same reasoning as in the previous case, this matrix corresponds to a coordination game ($A^{(ij)}=[A^{(ji)}]^\intercal$ and $A^{(jj)}$ is uniformly zero).
    
    	\begin{align}\tag{Coordination Game with Unique Nash for Odd $K$}
A=
	\left[
		\begin{array}{x{0.6cm} x{0.6cm} x{0.6cm} x{0.6cm} x{0.6cm} x{0.6cm} x{0.6cm}}
			\cline{1-3}
			\LL{0} & 0 & \RR{0} & \cellcolor{gray!25}1 & \cellcolor{gray!50}1 & 0& 0\\
			\cline{2-4}
			\LL{0} & \LL{0} & \RR{0} & \RR{0} & \cellcolor{gray!25}1 & \cellcolor{gray!50}1& 0\\
			\cline{3-5}
			\LL{0} & \LL{0} & \LR{0} & \RR{0} & \RR{0} & \cellcolor{gray!25}1& \cellcolor{gray!50}1\\
			\cline{1-3}\cline{4-6}
			\cellcolor{gray!50}1 & \LL{0} & \LR{0} & \RR{0} & \RR{0} & \RR{0}& \cellcolor{gray!25}1\\
			\cline{2-4}\cline{5-7}
			\cellcolor{gray!25}1 & \cellcolor{gray!50}1 & \LR{0} & \RR{0} & \RR{0} & \RR{0}& \RR{0}\\
			\cline{3-5}
			0 & \cellcolor{gray!25}1 & \RR{\cellcolor{gray!50}1} & \RR{0} & 0 & \RR{0}& \RR{0}\\
			\cline{4-6}
			0 & 0 & \cellcolor{gray!25}1 & \RR{\cellcolor{gray!50}1} & 0 & 0& \RR{0}\\
			\cline{5-7}
		\end{array}
	\right]
\end{align}

It remains to show that this game has a unique Nash equilibrium. 
Once again, this is equivalent to showing that $A$ has a non-zero determinant. 
Using the Leibniz formula for the determinant, $det(A)= \sum_{\sigma \in S_K} \prod_{w\in [K]} A_{w,\sigma_w}$, there are only two permutations of $[K]$ corresponding to a non-zero product.
The two permutations are described below and depicted by the light and dark gray cells in the matrix for $K=7$ above.
\begin{align*}
	\sigma^{(1)}&= \{\lceil K/2 \rceil , ...., K, 1 ,2 , ..., \lfloor K/2 \rfloor \}	\tag{$\sigma$ with Non-Zero Products}\\
	\sigma^{(2)}&= \{\lceil K/2 \rceil +1, ...., K, 1 ,2 , ..., \lfloor K/2 \rfloor +1 \}
\end{align*}

It is well known that the signature of a permutation $\sigma$ can be computed with $sgn(\sigma)= (-1)^{N(\sigma)}$ where $N(\sigma)$ is the number of inversions in $\sigma$, i.e., the number of $\{w,v\}$ where $w<v$ but $\sigma_w>\sigma_v$. 
For $\sigma^{(1)}= \{\lceil K/2 \rceil , ...., K, 1 ,2 , ..., \lfloor K/2 \rfloor \}$, each element of $\{1,...,\lfloor K/2 \rfloor\}$ is inverted with exactly $\lceil K/2 \rceil$ elements implying $N(\sigma^{(1)}) = \lceil K/2 \rceil \cdot  \lfloor K/2 \rfloor=\lceil K/2 \rceil\cdot (\lceil K/2 \rceil-1)$ since $K$ is odd. 
Therefore, $N(\sigma^{(1)})$ is even and $sgn(\sigma^{(1)})=1$. 
Through an identical argument $sgn(\sigma^{(2)})=1$ implying $det(A)=\sum_{\sigma \in \{\sigma^{(1)},\sigma^{(2)}\}}sgn(\sigma)\cdot\prod_{w \in \sigma}A_{w,\sigma_w}=\sum_{\sigma \in \{\sigma^{(1)},\sigma^{(2)}\}}sgn(\sigma)=2$ and $A$ has a unique Nash equilibrium.
Therefore, by Theorem \ref{thm:UniqueOrNone}, almost every coordination game has a unique Nash equilibrium when $K$ is odd.
This completes both cases and therefore almost every polymatrix coordination game has a unique Nash equilibrium.
\end{proof}

Since every coordination game is also a general game, the proof for general games follows almost identically.  
We only need to redefine our matrix function and the open, connected set generating games.

\begin{theorem}\label{thm:GPerturb}
    Almost every polymatrix (general) game has a unique Nash equilibrium if and only if $k_i\leq K/2=\sum_{j\in [n]} k_j/2$ for all $i\in [n]$.
    Formally, let ${\cal U}^{(ij)}=\mathbb{R}^{k_i\times k_j}$, ${\cal U}=\times_{i\neq j}{\cal U}^{(ij)}$, $A^{(ij)}(U)=U^{(ij)}$, and let $\mu$ be an arbitrary full-dimensional Gaussian distribution over $\cal U$. 
    Then the set of $U\in {\cal U}$ that has a unique Nash equilibrium has probability measure one if and only if $k_i\leq K/2=\sum_{j\in [n]} k_j/2$ for all $i\in [n]$. 
\end{theorem}


\begin{proof}
    As in the proof of Theorem \ref{thm:CPerturb}, $(A,{\cal U}, \mu)$ captures all general polymatrix games. 
    Further, since $A$ is analytic, and ${\cal U}$ is connected and open, Theorem \ref{thm:UniqueOrNone} applies and it suffices to give a single instance where there is a unique equilibrium. 
    Since coordination games are a subset of zero-sum games, the instances in the proof of Theorem \ref{thm:CPerturb} suffice to complete the proof.
\end{proof}

\subsection{Zero-Sum Games}\label{sec:ZS}

While  zero-sum games are a subset of general games, they are a measure-zero subset and therefore the conditions for general games do not necessarily extend to zero-sum games. 
In contrast, many zero-sum games will fail to have unique equilibria;
we establish that the existence of a unique Nash equilibria is a general condition in polymatrix zero-sum games if and only if (1) the dimension of each agent's strategy space is at most half of the dimension of the combined strategy space and (2) the total number of strategies is even.

\begin{proposition}\label{thm:ZSEven}
    The polymatrix zero-sum game $G=(n,k,A,b)$ does not have a unique Nash equilibrium if $K=\sum_{i\in [n]} k_i$ is odd.
\end{proposition}
\begin{proof}
	$det(A)=det(A^\intercal)=det(-A)=(-1)^{K}det(A)=-det(A)$ since $K$ is odd, and $det(A)=0$. 
	Thus $A$ is not invertible implying no unique Nash equilibrium. 
\end{proof}

Despite the simplicity of the proof, the result is quite surprising --- a polymatrix zero-sum game cannot have a unique equilibrium if the dimension of the strategy space is odd.  
This means that symmetric games with an odd number of agents strategies will not have a unique equilibrium. 
We show that this condition on the parity of $K$ is also sufficient for unique equilibria.

\begin{theorem}\label{thm:ZSPerturb}
    Almost every zero-sum game has a unique Nash equilibrium if and only if 
    \begin{itemize}
        \item $K=\sum_{i\in [n]}k_i$ is even,
        \item and $k_i\leq K/2$ for all $i\in [n]$.
    \end{itemize}
    Formally, let ${\cal U}^{(ij)}=\mathbb{R}^{k_i\times k_j}$ for $i<j$, ${\cal U}=\times_{i<j}{\cal U}^{(ij)}$.  
    For $i<j$, let $A^{(ij)}(U)=U^{(ij)}$ and $A^{(ji)}(U)=-[U^{(ij)}]^\intercal$. 
    Finally, let $\mu$ be an arbitrary full-dimensional Gaussian distribution over $\cal U$. 
    Then the set of $U\in {\cal U}$ that has a unique Nash equilibrium has probability measure one if and only if the above conditions are satisfied. 
\end{theorem}

\begin{proof}
    The necessary conditions are given by Propositions \ref{thm:Half} and \ref{thm:ZSEven}. 
    The remainder of the proof follows identically to the proofs of Theorems \ref{thm:CPerturb} and \ref{thm:GPerturb}. 
    The selection of $(A,{\cal U}, \mu)$ captures all zero-sum polymatrix games. 
    Further, since $A$ is analytic, and ${\cal U}$ is connected and open, Theorem \ref{thm:UniqueOrNone} applies and it suffices to give a single instance where there is a unique equilibrium.

   The game we construct is nearly identically to the coordination game we constructed when $K$ is even. 
    Let $A_{w,\frac{K}{2}+w}=1$ and $A_{\frac{K}{2}+w,w}=-1$ for all $w\in \left[\frac{K}{2}\right]$ and let $A_{w,v}=0$ for all other pairs $\{w,v\}$ implying $A=-A^\intercal$. 
    The only difference between this instance and the even instance for coordination games is that the subdiagonal entries are now $-1$ ensuring that $A=-A^\intercal$. 

    The remainder of the proof then follows identically to the proof of Theorem \ref{thm:CPerturb} --- the determinant is non-zero implying a unique Nash equilibrium and corresponds.
    Further, as in the proof of Theorem \ref{thm:CPerturb}, $A$ is constructed so that it corresponds to a polymatrix zero-sum game where agent $i$ controls $k_i$ strategies.  
    Thus, almost every zero-sum game where $K$ is even and $k_i\leq K/2$ for all $i\in [n]$ will have a unique Nash equilibrium. 
\end{proof}

\subsection{Games with Constrained Strategy Spaces}\label{sec:Constraints}

In this section, we briefly discuss what our results indicate about games with constrained strategy spaces, e.g., when strategies correspond to probability vectors.

Consider an agent that has their strategies confined to the affine space $a^\intercal x_i = c$.
We may perform a variable substitution for $x_{iw}$ for a $w$ where $a_w\neq 0$, i.e., we replace $x_{iw}$ with $(c-a_{-w}x_{i,-w})/a_w$ where the $-w$ is used to indicate every element of the vector except for the $w$th element. 
Notably, when performing this transformation if the original game has form $x_j^\intercal A^{(ji)} x_i$, then the new game has the form $x_j^\intercal A^{(ji)} [x_{i,-w}, (c-a_{-w}x_{i,-w})/a_w]=x_j^\intercal \bar{A}^{(ji)}x_{i,-w} + \bar{b}_j x_j$ for some new $\bar{A}$ and $\bar{b}$ where $\bar{b}_j$ corresponds to a linear cost for agent $j$. 
Thus, all of our results extend to the setting when strategies space is affine. 

However, the problem is more complicated when the strategy space is constrained by inequalities.  Consider the case when strategies are probability vectors which require the affine constraint $\mathbf{1}^\intercal x_i =1$ and the half spaces $x_{ik}\geq 0$ for all $k$. 
For the set of games where there is a Nash equilibrium in the relative interior of the strategy space (equivalently, a fully-mixed equilibrium), our results extend since the half-spaces are not tight and may be ignored.  
E.g., by Proposition \ref{thm:ZSEven} there is no 3-agent game where each agent controls 2 strategies (a one-dimensional space) that has a unique, fully-mixed Nash equilibrium since the dimension of the combined strategy space is odd. 
However, if the smallest facet containing the set of Nash equilibria is 2-dimensional (e.g., when exactly one strategy is dominated), then almost every game will have a unique equilibrium. 
Similarly, it is possible that a game with even dimension may have a Nash equilibrium with odd support, implying there will be multiple equilibria.  
This means that the lack of uniqueness becomes a more significant problem when examining games with constrained strategy spaces. 

\section{Recovering Uniqueness with Learning Dynamics}\label{sec:Recover}

In this section, we examine the impact of non-uniqueness on the classical learning dynamic \ref{eqn:ContGD}.  
We remark that our results likely extend to discrete learning algorithms with have high regularity, e.g., alternating gradient descent approximately preserves the dynamics of \ref{eqn:ContGD} \cite{Bailey2019Regret,Bailey19Hamiltonian}.
However, it likely that different families of algorithms may require different approaches to verify whether they converge to a unique solution. 

Online learning in polymatrix zero-sum games has received countless attention due to its application in areas such as Generative Adversarial Neural Networks \cite{goodfellow2014generative}, bargain and resource allocation problems \cite{Shahrampour20OnlineAllocation}, and policy evaluation methods \cite{du2017stochastic}.
There are many algorithms and learning dynamics that guarantee relatively fast convergence to the set of Nash equilibria. 
For instance, \ref{eqn:ContGD} has $O(1/T)$ time-average convergence (informally, $\int_0^t x(s)ds/t\to x^*$) to the set Nash equilibria given an arbitrary starting strategy profiles $x(0)$ \cite{Mertikopoulos2018CyclesAdverserial}. 
\begin{align*}
    x_i(t) = x_i(0)+ \int_0^t \left( -b_i + \sum_{j\neq i} A^{(ij)}x_j(s)\right) ds \tag{Continuous-time Gradient Descent}\label{eqn:ContGD}
\end{align*}

In two-agent zero-sum games, it is well-known that almost every game in a compact space yields a unique Nash equilibrium \cite{van1991stability}. 
Therefore these fast learning algorithms and dynamics provide a simple, unique solution for agents' actions in these two-agent games. 
However, standard approaches for proving convergence only establish that the learning algorithms converge (or achieve time-average convergence) to the \textit{set} of Nash equilibria, which we've established is frequently not a singleton in the multiagent setting.  
This raises natural questions about the actual behavior of learning dynamics when the equilibrium is not unique: 
\textit{Which Nash equilibrium does our learning algorithm converge to? Does it converge to a single Nash equilibrium at all or does it wander between different equilibria?}
In this section, we establish that the classical learning dynamic gradient descent converges (time-average) to a unique Nash equilibrium based on the starting strategy profile; specifically we show that the time-average of the dynamics converge to the Nash equilibrium closest to the initial strategy profile $x(0)$. 

\subsection{A Cursory Understanding of \ref{eqn:ContGD}}\label{sec:Dynamics}

We begin by showing that \ref{eqn:ContGD} converges to the set of Nash equilibria. 
We remark that convergence is informally known to be a natural consequence of Lemma \ref{lem:invariance}, but, to our knowledge, no published paper establishes the result. 

\begin{lemma}[Invariant energy \cite{Mertikopoulos2018CyclesAdverserial,Bailey19Hamiltonian}]\label{lem:invariance}
    Let $x^*$ be an arbitrary Nash equilibrium for the polymatrix game $G=(n,k,A,b)$. 
    If the initial strategy profile $x(0)$ is updated with \ref{eqn:ContGD}, then  $||x(0)-x^*||^2=||x(t)-x^*||^2$ for all $t\geq 0$. 
\end{lemma}


    


\begin{theorem}\label{thm:converge1}
    Let $x^*$ be an arbitrary Nash equilibrium for the polymatrix game $G=(n,k,A,b)$. 
    If the initial strategy profile $x(0)$ is updated with \ref{eqn:ContGD} then the time-average of the strategies $\bar{x}(t)=\int_0^tx(s)ds$ converges to the set of Nash equilibria. 
\end{theorem}
\begin{proof}
    The set of Nash equilibria is given by $\{x: Ax=b\}$ and it suffices to show $\lim_{t\to \infty} \int_0^t Ax(s)ds/t\to b$.

    First, $x(t)=x(0)+\int_0^t (-b+Ax(s))ds$ and therefore $\int_0^t Ax(s)ds=x(t)-x(0)+t\cdot b$.
    Finally, by Lemma \ref{lem:invariance}, $x(t)-x(0)$ is bounded and therefore $(x(t)-x(0))/t\to \vec{0}$ as $t\to \infty$ yielding time-average convergence to the set of Nash equilibria. 
    \end{proof}

\subsection{Converging to the Closest Nash Equilibrium.}\label{sec:reduction}

Lemma \ref{lem:invariance} and Theorem \ref{thm:converge1} provide everything we need to establish convergence to the (unique) Nash equilibrium that is closest to initial strategy profile $x(0)$. 
The key part of Lemma \ref{lem:invariance} is that the strategies remain equidistant from \textit{every} Nash equilibrium. 
As depicted in Figure \ref{fig:MultiEquil} on page \pageref{fig:MultiEquil}, this means that, in the event of two equilibria, that the strategies will always be contained on two distinct hyperspheres implying the strategies will actually be contained to the hyperplane defined by the intersection of these two hyperspheres (Proposition \ref{thm:hyperplane}). 
Further, since the strategies are contained to this hyperplane, so is the time-average of the strategies.  
We show that there is a unique Nash equilibrium $x^*$ contained in these hyperplanes and therefore the only way we can achieve time-average convergence to the set of Nash equilibria (within these hyperplanes) is by converging directly to $x^*$ (Theorem \ref{thm:converge2}).

\begin{proposition}\label{thm:hyperplane}
    Suppose agents' strategies are updated with \ref{eqn:ContGD} in a polymatrix zero-sum game $G=(n,k,A,b)$ starting with the initial strategy profile $x(0)$.  
    Let $x^*$ be the unique Nash equilibrium closest to $x(0)$. 
    If $x^*+d$ is also a Nash equilibrium, then $d^\intercal x(t)=d^\intercal x^*$ for all $t$. 
\end{proposition}
\begin{proof}
    First observe that the set of Nash equilibria is affine, i.e., 
    if $x^*$ and $x^*+d$ are Nash equilibria, then $x^*+\lambda\cdot d$ is a Nash equilibrium for all $\lambda\in \mathbb{R}$:
    By Lemma \ref{lem:Nash}, $x$ is a Nash equilibrium if and only if $Ax=b$.
    Therefore $Ad=A(x^*+d)-Ax^*=0$ and $A(x^*+\lambda\cdot d)= Ax^*+\lambda\cdot Ad=Ax^*=b$ implying $x^*+\lambda\cdot d$ is also a Nash equilibrium.

    Since the set of Nash equilibria is affine (and therefore convex), there is a unique solution to $\min_{x\in Nash} ||x_0-x||$ and $x^*$ is well-defined. 
    
    Since $x^*$ is the closest Nash equilibrium to $x(0)$, the points $x(0)$, $x^*$, and $x^*\pm d$ form two right angles and $x^*-d$ and $x^*+d$ are equidistant from $x(0)$;
    by the Pythagorean theorem,
        $ ||x(0)-(x^*\pm d)||^2
        = ||x(0)-x^*||^2 + ||x^*-(x^*\pm d)||^2
        = ||x(0)-x^*||^2 + ||d||^2.$
    
    Further, by Lemma \ref{lem:invariance}, the distance to each Nash equilibrium is invariant and therefore $||x(t)-(x^*+d)||^2=||x(t)-(x^*-d)||^2$ at every time $t\geq 0$. 
    Expanding and re-arranging terms yields the equality $d^\intercal x(t)=d^\intercal x^*$. 
\end{proof}

\begin{theorem}\label{thm:converge2}
    Suppose agents' strategies are updated with \ref{eqn:ContGD} in a polymatrix zero-sum game $(A,b)$ starting with the initial strategy profile $x(0)$.  
    Let $x^*$ be the unique Nash equilibrium closest to $x(0)$. 
    Then the time-average of the strategies $\bar{x}(t)=\int_0^tx(s)ds$ converges to $x^*$.
\end{theorem}

\begin{proof}
    If there is a unique Nash equilibrium, the result trivially holds by Theorem \ref{thm:converge1}.
    Next, suppose the set of Nash equilibria is given by $\{x= x^*+\sum_{w=1}^W \lambda_w\cdot d_w: \lambda\in \mathbb{R}^W\}$ where $\{d_w\}_{w=1}^W$ is a basis for the set of Nash equilibria. 
    By Theorem \ref{thm:converge1}, $
        \lim_{t\to \infty} \bar{x}(t)\subseteq \{x^*+\sum_{w=1}^W \lambda_w\cdot d_w: \lambda\in \mathbb{R}^W\}$ --- here we use $\lim$ for the limit set since Theorem \ref{thm:converge1} does not guarantee the limit exists.
        

    By Proposition \ref{thm:hyperplane}, $x(t)\in \{x: d_w^\intercal x = d_w^\intercal x^*\}$ for $w=1,...,W$ and therefore so is $\bar{x}(t)$.
    Thus, $
        \lim_{t\to \infty}\bar{x}(t)\subseteq \bigcap_{w=1}^W\{x: d_w^\intercal x=d_w^\intercal x^*\}$ 
    since hyperplanes are closed, i.e., since the limit point of a sequence of  points in the hyperplane is also in the hyperplane.
    Observe that $x^*$ is the only Nash equilibrium in $\bigcap_{w=1}^W\{x: d_w^\intercal x=d_w^\intercal x^*\}$ since any other Nash equilibrium in the form $x^*+\sum_{w=1}^W \lambda_w\cdot d_w$ violates $d_w^\intercal x= d_w^\intercal x^*$ for any $w$ where $\lambda_w\neq 0$.
    Therefore, $
         \bigcap_{w=1}^W\{x: d_w^\intercal x=d_w^\intercal x^*\}\cap \{x= x^*+\sum_{w=1}^W \lambda\cdot d_w: \lambda\in \mathbb{R}^W\}=\{x^*\}$
     and $\lim_{t\to \infty} \bar{x}(t)=x^*$.
\end{proof}

\section{Conclusions}
We provide necessary and sufficient conditions for uniqueness of Nash equilibria in zero-sum, coordination, and general polymatrix games. 
We show there are larges classes of polymatrix zero-sum games that fail to yield a unique equilibrium. 
While this result may be undesirable in many settings, we also show how the multiplicity of Nash equilibria can improve our understanding of classical multiagent learning dynamics. 
 
\bibliographystyle{plain}
\bibliography{References}

\begin{thebibliography}{10}

\bibitem{abernethy2021last}
Jacob Abernethy, Kevin~A Lai, and Andre Wibisono.
\newblock Last-iterate convergence rates for min-max optimization: Convergence of hamiltonian gradient descent and consensus optimization.
\newblock In {\em Algorithmic Learning Theory}, pages 3--47. PMLR, 2021.

\bibitem{Aumann87}
Robert~J. Aumann.
\newblock Game theory.
\newblock In {\em The New Palgrave: A Dictionary of Economics by J. Eatwell, M. Milgate, and P. Newman}, page 460–482. Macmillan, 1987.

\bibitem{Bailey2019Regret}
James~P Bailey, Gauthier Gidel, and Georgios Piliouras.
\newblock Finite regret and cycles with fixed step-size via alternating gradient descent-ascent.
\newblock In {\em Conference on Learning Theory}, pages 391--407. PMLR, 2020.

\bibitem{bailey2021stochastic}
James~P Bailey, Sai~Ganesh Nagarajan, and Georgios Piliouras.
\newblock Stochastic multiplicative weights updates in zero-sum games.
\newblock {\em arXiv preprint arXiv:2110.02134}, 2021.

\bibitem{Bailey18Divergence}
James~P. Bailey and Georgios Piliouras.
\newblock Multiplicative weights update in zero-sum games.
\newblock In {\em Proceedings of the 2018 ACM Conference on Economics and Computation}, EC ’18, page 321–338, New York, NY, USA, 2018. Association for Computing Machinery.

\bibitem{Bailey19GDRegret}
James~P. Bailey and Georgios Piliouras.
\newblock Fast and furious learning in zero-sum games: Vanishing regret with non-vanishing step sizes.
\newblock In {\em Advances in Neural Information Processing Systems 32}, pages 12977--12987. Curran Associates, Inc., 2019.

\bibitem{Bailey19Hamiltonian}
James~P. Bailey and Georgios Piliouras.
\newblock Multi-agent learning in network zero-sum games is a hamiltonian system.
\newblock In Edith Elkind, Manuela Veloso, Noa Agmon, and Matthew~E. Taylor, editors, {\em Proceedings of the 18th International Conference on Autonomous Agents and MultiAgent Systems, {AAMAS} '19, Montreal, QC, Canada, May 13-17, 2019}, pages 233--241. International Foundation for Autonomous Agents and Multiagent Systems, 2019.

\bibitem{barelli2013note}
Paulo Barelli and Idione Meneghel.
\newblock A note on the equilibrium existence problem in discontinuous games.
\newblock {\em Econometrica}, 81(2):813--824, 2013.

\bibitem{cai2016zero}
Yang Cai, Ozan Candogan, Constantinos Daskalakis, and Christos Papadimitriou.
\newblock Zero-sum polymatrix games: A generalization of minmax.
\newblock {\em Mathematics of Operations Research}, 41(2):648--655, 2016.

\bibitem{carbonell2018existence}
Oriol Carbonell-Nicolau and Richard~P McLean.
\newblock On the existence of nash equilibrium in bayesian games.
\newblock {\em Mathematics of Operations Research}, 43(1):100--129, 2018.

\bibitem{dasgupta1986existence}
Partha Dasgupta and Eric Maskin.
\newblock The existence of equilibrium in discontinuous economic games, i: Theory.
\newblock {\em The Review of economic studies}, 53(1):1--26, 1986.

\bibitem{daskalakis2019last}
C~Daskalakis and Ioannis Panageas.
\newblock Last-iterate convergence: Zero-sum games and constrained min-max optimization.
\newblock In {\em 10th Innovations in Theoretical Computer Science (ITCS) conference, ITCS 2019}, 2019.

\bibitem{daskalakis2021nearoptimal}
Constantinos Daskalakis, Maxwell Fishelson, and Noah Golowich.
\newblock Near-optimal no-regret learning in general games, 2021.

\bibitem{deligkas2017computing}
Argyrios Deligkas, John Fearnley, Rahul Savani, and Paul Spirakis.
\newblock Computing approximate nash equilibria in polymatrix games.
\newblock {\em Algorithmica}, 77(2):487--514, 2017.

\bibitem{du2017stochastic}
Simon~S Du, Jianshu Chen, Lihong Li, Lin Xiao, and Dengyong Zhou.
\newblock Stochastic variance reduction methods for policy evaluation.
\newblock In {\em International Conference on Machine Learning}, pages 1049--1058. PMLR, 2017.

\bibitem{fu2022equilibrium}
Qiang Fu, Zenan Wu, and Yuxuan Zhu.
\newblock On equilibrium existence in generalized multi-prize nested lottery contests.
\newblock {\em Journal of Economic Theory}, 200:105377, 2022.

\bibitem{Gidel2019Momentum}
Gauthier Gidel, Reyhane~Askari Hemmat, Mohammad Pezeshki, R\'emi~Le Priol, Gabriel Huang, Simon Lacoste-Julien, and Ioannis Mitliagkas.
\newblock Negative momentum for improved game dynamics.
\newblock In Kamalika Chaudhuri and Masashi Sugiyama, editors, {\em Proceedings of Machine Learning Research}, volume~89 of {\em Proceedings of Machine Learning Research}, pages 1802--1811. PMLR, 16--18 Apr 2019.

\bibitem{goodfellow2014generative}
Ian~J. Goodfellow, Jean Pouget-Abadie, Mehdi Mirza, Bing Xu, David Warde-Farley, Sherjil Ozair, Aaron Courville, and Yoshua Bengio.
\newblock Generative adversarial networks, 2014.

\bibitem{howson1972equilibria}
Joseph~T Howson~Jr.
\newblock Equilibria of polymatrix games.
\newblock {\em Management Science}, 18(5-part-1):312--318, 1972.

\bibitem{Mertikopoulos2018CyclesAdverserial}
Panayotis Mertikopoulos, Christos Papadimitriou, and Georgios Piliouras.
\newblock Cycles in adversarial regularized learning.
\newblock In {\em Proceedings of the Twenty-Ninth Annual ACM-SIAM Symposium on Discrete Algorithms}, SODA ’18, page 2703–2717, USA, 2018. Society for Industrial and Applied Mathematics.

\bibitem{mityagin2015zero}
Boris Mityagin.
\newblock The zero set of a real analytic function.
\newblock {\em arXiv preprint arXiv:1512.07276}, 2015.

\bibitem{miura2008security}
R~Ann Miura-Ko, Benjamin Yolken, John Mitchell, and Nicholas Bambos.
\newblock Security decision-making among interdependent organizations.
\newblock In {\em 2008 21st IEEE Computer Security Foundations Symposium}, pages 66--80. IEEE, 2008.

\bibitem{naghizadeh2017provision}
Parinaz Naghizadeh and Mingyan Liu.
\newblock Provision of public goods on networks: on existence, uniqueness, and centralities.
\newblock {\em IEEE Transactions on Network Science and Engineering}, 5(3):225--236, 2017.

\bibitem{Shahrampour20OnlineAllocation}
S.~{Pu}, J.~J. {Escudero-Garzas}, A.~{Garcia}, and S.~{Shahrampour}.
\newblock An online mechanism for resource allocation in networks.
\newblock {\em IEEE Transactions on Control of Network Systems}, pages 1--1, 2020.

\bibitem{toonsi2024higher}
Sarah Toonsi and Jeff Shamma.
\newblock Higher-order uncoupled dynamics do not lead to nash equilibrium-except when they do.
\newblock {\em Advances in Neural Information Processing Systems}, 36, 2024.

\bibitem{van1991stability}
Eric Van~Damme.
\newblock {\em Stability and perfection of Nash equilibria}, volume 339.
\newblock Springer, 1991.

\bibitem{pmlr-v202-vyas23a}
Abhijeet Vyas, Brian Bullins, and Kamyar Azizzadenesheli.
\newblock Competitive gradient optimization.
\newblock In Andreas Krause, Emma Brunskill, Kyunghyun Cho, Barbara Engelhardt, Sivan Sabato, and Jonathan Scarlett, editors, {\em Proceedings of the 40th International Conference on Machine Learning}, volume 202 of {\em Proceedings of Machine Learning Research}, pages 35243--35276. PMLR, 23--29 Jul 2023.

\bibitem{NEURIPS2023_bf331c87}
Guillaume Wang and L\'{e}na\"{\i}c Chizat.
\newblock Local convergence of gradient methods for min-max games: Partial curvature generically suffices.
\newblock In A.~Oh, T.~Naumann, A.~Globerson, K.~Saenko, M.~Hardt, and S.~Levine, editors, {\em Advances in Neural Information Processing Systems}, volume~36, pages 60841--60852. Curran Associates, Inc., 2023.

\end{thebibliography}

\end{document}